\documentclass[a4paper,12pt,oneside]{article}
\usepackage{amsmath,amssymb,mathtools,bm,lipsum}
\usepackage[utf8]{inputenc}
\usepackage[english]{babel}
\usepackage[T1]{fontenc}
\usepackage{amsthm}
\usepackage{bbm}
\usepackage{enumerate}
\usepackage{hyperref}
\usepackage[capitalise]{cleveref}
\usepackage{braket}
\usepackage{dsfont}

\usepackage[verbose=true,letterpaper]{geometry}
\AtBeginDocument{
	\newgeometry{
		textheight=9.5in,
		textwidth=7in,
		top=1in,
		headheight=14.5pt,
		headsep=10pt,
		footskip=20pt
	}
}

\usepackage{fancyhdr}
\fancyheadoffset{0pt}
\usepackage{titlesec}
\titleformat{\section}{\bfseries\scshape\Large}{\thesection}{1em}{}{}
\titleformat*{\subsection}{\scshape\bfseries\large}

\renewenvironment{thebibliography}[1]{
	\begin{oldthebibliography}{#1}
		\footnotesize
		\setlength{\itemsep}{0em}
		\setlength{\parskip}{0em}
	}
	{
	\end{oldthebibliography}
}

\newcommand{\IN}{\mathbb{N}}
\newcommand{\IZ}{\mathbb{Z}}
\newcommand{\IR}{\mathbb{R}}
\newcommand{\IC}{\mathbb{C}}
\newcommand{\R}{\mathbb{R}}

\newcommand{\Z}{\mathbb{Z}}
\newcommand{\N}{\mathbb{N}}

\newcommand{\EE}{\mathbb{E}}

\newcommand{\PP}{\mathbb{P}}

\newcommand{\FS}{\mathcal{F}}

\newcommand{\cA}{\mathcal{A}}

\newcommand{\cD}{\mathcal{D}}\newcommand{\cP}{\mathcal{P}}

\newcommand{\cS}{\mathcal{S}}

\newcommand{\cI}{\mathcal{I}}

\newcommand{\cZ}{\mathcal{Z}} 

\newcommand{\eps}{\varepsilon}
\newcommand{\ph}{\varphi}

\newcommand{\dG}{\mathsf{d}\Gamma}

\numberwithin{equation}{section}
\newtheorem{theorem}{Theorem}[section]
\newtheorem{lemma}[theorem]{Lemma}
\newtheorem{proposition}[theorem]{Proposition}
\newtheorem{cor}[theorem]{Corollary}

\theoremstyle{definition}

\theoremstyle{remark}
\newtheorem{remark}[theorem]{Remark}

\crefname{hyp}{Hypothesis}{Hypotheses}
\Crefname{hyp}{Hypothesis}{Hypotheses}
\crefname{lemma}{Lemma}{Lemmas}
\Crefname{lemma}{Lemma}{Lemmas}
\crefname{enumi}{}{}
\Crefname{enumi}{}{}
\creflabelformat{enumi}{#2(#1)#3}
\crefname{equation}{}{}
\Crefname{equation}{}{}

\usepackage{tikz}

\def\titlename{\scshape  Correlation Bound for a  One-Dimensional Continuous Long-Range Ising Model}
\title{\LARGE\scshape  Correlation Bound  for a  One-Dimensional  Continuous Long-Range Ising Model}
\author{David Hasler\footnote{\texttt{david.hasler@uni-jena.de}} \qquad Benjamin Hinrichs\footnote{\texttt{benjamin.hinrichs@uni-jena.de}} \qquad Oliver Siebert\footnote{Present Affiliation: 
		\'Ecole polytechnique f\'ed\'erale de Lausanne, \texttt{oliver.siebert@epfl.ch}}\\ Friedrich-Schiller-University Jena\\\small Department of Mathematics\\[-.5em]\small   Ernst-Abbe-Platz 2\\[-.5em]\small  07743 Jena\\[-.5em] \small Germany}
\newcommand{\shortauthors}{D. Hasler, B. Hinrichs, O. Siebert}

\newcommand{\Id}{{\mathds{1}}}

\newcommand{\wde}{{ w}^{(\delta)}}

\renewcommand{\EE}{\mathds{E}}
\newcommand{\bL}[1]{\braket{#1}^{(L)}}

\newcommand{\bLe}[1]{\braket{#1}^{(L,\eta)}}

\newcommand{\bn}[1]{\llangle{#1}\rrangle^{\mathrm{(n)}}}
\newcommand{\Bn}[1]{\left\llangle{#1}\right\rrangle^{\mathrm{(n)}}}
\newcommand{\bw}[1]{\llangle#1\rrangle}

\newcommand{\id}{{ \mathfrak{i}_\delta}}
\newcommand{\Ld}{{ L_\delta}}
\newcommand{\jd}{{ j_\delta}}

\newcommand{\sd}[1]{\sigma_{\id(#1)}}
\newcommand{\fd}{{\mathfrak s_\delta}}

\DeclareFontFamily{OMX}{MnSymbolE}{}
\DeclareFontShape{OMX}{MnSymbolE}{m}{n}{
	<-6>  MnSymbolE5
	<6-7>  MnSymbolE6
	<7-8>  MnSymbolE7
	<8-9>  MnSymbolE8
	<9-10> MnSymbolE9
	<10-12> MnSymbolE10
	<12->   MnSymbolE12}{}
\DeclareSymbolFont{mnlargesymbols}{OMX}{MnSymbolE}{m}{n}
\SetSymbolFont{mnlargesymbols}{bold}{OMX}{MnSymbolE}{b}{n}
\DeclareMathDelimiter{\llangle}{\mathopen}{mnlargesymbols}{'164}{mnlargesymbols}{'164}
\DeclareMathDelimiter{\rrangle}{\mathclose}{mnlargesymbols}{'171}{mnlargesymbols}{'171}

\begin{document}
	
	\maketitle\thispagestyle{empty}
	\begin{abstract}
	\noindent
	We consider a measure  given as  the   continuum limit of a one-dimensional Ising model with long-range translationally invariant  interactions.
	Mathematically, the   measure can be described by a self-interacting  Poisson driven jump process. We prove a     correlation inequality,
estimating  the  magnetic susceptibility of  this model, which holds   for  small $L^1$-norm of the interaction function.
The bound on the magnetic susceptibility has applications in quantum field theory and  can be used to prove existence
of ground states for the spin boson model.
	\end{abstract}

\section{Introduction and Result}\label{sec:introandresult}

The spin boson model describes a two-level quantum mechanical system linearly coupled to a quantized bosonic field. If the bosons have relativistic dispersion
relation the spin boson model provides a  caricature of a confined non-relativistic quantum
mechanial system interacting with the quantized electromagnetic field.
This is one of the reasons the  spin boson model has been intensively investigated
and questions about existence of a ground state of the spin boson Hamiltonian
are of interest.
In a recent paper \cite{HaslerHinrichsSiebert.2021a}, we showed that the
spin boson model has a square integrable  ground state in situtations where the coupling function
can have infrared singularities as long as the Hamiltonian stays bounded from below.
The result assumed a resolvent bound.
It is well-known that the ground state energy of the spin boson model
can be expressed in terms of an expectation of a one-dimensional continuous
Ising model with long range couplings  \cite{EmeryLuther.1974}.
In this paper we prove a bound  about this continuous Ising model,
which can be used to obtain the resolvent bound needed in  \cite{HaslerHinrichsSiebert.2021a}.

The Ising model is a mathematical model of ferromagnetism  and has been intensively investigated. The magnetic dipole moments are
approximated by the values $\{ + 1 , - 1 \}$  often referred to as Ising spins.
The Ising spins are typically arranged on a lattice. In this paper
we consider a  one-dimensional continuous Ising model which
is described in terms of a jump process and a long range  interaction given by  a nonnegative
symmetric integrable function. The main result of this paper is a correlation bound, which  has the
the physical interpretation of a bound on the magnetic  susceptibility. Thus, the bound
which we prove is of its own physical interest. In fact, to prove the bound we will consider a scaling limit
of an Ising model on the one-dimensional lattice, where the nearest neighbor coupling becomes arbitrarily
large. Thus to  obtain the  main result we will prove   a correlation  bound for  the   Ising model
on $\Z$. That bound is  of its own  interest  and can be viewed as a result  between   the  results of  Dyson  \cite{Dyson.1969} and of Rogers and Thompson
\cite{RogersThompson.1981}  for  the  Ising model on the one-dimensional lattice. 

Let us now give the explicit definition. Let $N(t)$ with $t\in\IR$ be a two-sided Poisson process with unit intensity and let $B$ be an independent Bernoulli random variable with $\PP(B=1)=\PP(B=-1)=\frac 12$. Then define $X(t)$ to be the jump process
\begin{equation}\label{def:X}
	X(t) = B(-1)^{N(t)}.
\end{equation}
We give an overview of the connection between the jump process and the spin boson model, as motivated in the beginning of this introduction, in \cref{sec:spinboson}. This allows us to describe the desired second derivative of the ground state energy as an expectation value (cf. \cref{eq:seconddermu0}).

To state the main result of this paper, assume $W:\IR\to\IR$ is continuous. We then define the canonical partition function
\begin{equation}\label{def:partitionjump}
	\cZ(W,T) = \EE\left[\exp\left(\int_{-T}^T\int_{-T}^T W(t-s)X(t)X(s)dtds\right)\right] \qquad\mbox{for}\ T>0.
\end{equation}
\begin{remark}
The integral occuring in \eqref{def:partitionjump} is a Riemann integral. If the jump process is realized as a random variable on a space $\Omega$,
then for almost every $\omega\in\Omega$ the function  $t \mapsto  X(t)(\omega)$ has only finitely many discontinuities  on compact subsets and is hence
Riemann integrable.
\end{remark}
Our central result is the following.
\begin{theorem} \label{main}  There exist constants $\eps>0$ and $C> 0$, such that for all continuous and even $W \in L^1(\IR)$ with $W \geq 0$ and $\| W \|_1 \leq \eps$, we have
	\[\limsup_{T \to \infty}  \frac{1}{\cZ(W,T)}
	\EE \left[  \frac 1T  \left( \int_{-T}^T X(t) dt \right)^2  \exp\left( \int_{-T}^T \int_{-T}^T X(t) X(s) W(t-s) dt ds \right)  \right]  \le C.\]
\end{theorem}
\begin{remark} This result equivalently holds, if we choose an arbitrary intensity $\lambda$ of the Poisson process $N$ in \cref{def:X}. Note that the constant $C$, however, is not independent of $\lambda$. This can be seen by a simple scaling argument.
\end{remark}
\begin{remark}
In the special  case, where the integrable  $W \geq 0$ satisfies the additional condition
$W(t) \sim t^{-2}$ as $ t \to \infty$, a  bound as in \cref{main}    follows from  \cite[Proposition 8.1]{Spohn.1989}  for the conditioned process with boundary conditions $X(T)= X(-T)$.
 The proof  given in  \cite{Spohn.1989}  is based on results from  percolation theory  \cite{AizenmanNewman.1986}.  
\end{remark}

\begin{remark}
The bound in \cref{main}  is in general not expected to   hold for arbitrary large $\eps > 0$, as the following results indicate.
 The one-dimensional long-range Ising model with spins $\sigma_i = \pm 1$, $i \in \Z$ and  interaction energy   $\sum_{i,j} J(i-j) \sigma_i \sigma_j $ with
$J(n) = n^{-\alpha}$  has a phase transition  if  $1 < \alpha \leq 2$. In that case   the  magnetic susceptibility diverges for sufficiently small
temperatures. This was shown in \cite{Dyson.1969} for $1 < \alpha < 2$ and in \cite{AizenmanChayesChayesNewman.1980}
for $\alpha = 2$. It is reasonable to believe that such a divergence  carries over to the continuous  model, since the continuous
model can be obtained by a scaling limit of the discrete model if  an additional nearest neighbor coupling is imposed.
For details on the scaling limit,   we refer the reader to  \cref{sec:continuumlimit} and also  \cite{SpohnDuemcke.1985,Spohn.1989}.
\end{remark}

The paper is organized as follows. In \cref{sec:spinboson}, we describe the connection between the jump process $X$ and the spin boson model. This illustrates our motivation to study the problem, but is not relevant for the proof of the result. In \cref{sec:cor}, we prove an upper bound for correlation functions in the one-dimensional Ising model on $\IZ$. This is the main technical ingredient of our proof and can also be seen in a long line of such estimates, as described in the beginning of that section. In \cref{sec:continuumlimit}, we then prove the jump process is the continuum limit of a specific Ising model. We summarize our proof of \cref{main} in \cref{sec:proof}.

\section{Magnetic Susceptibility and   the Spin Boson Model}\label{sec:spinboson}

This section is  not needed for the proof of   \cref{main}, but rather puts the result into a broader context.
 First, we  show  that the expression in \cref{main}
is equal  to the magnetic susceptibility. Then, we relate the result to the spin boson model. In particular, we sketch
how \cref{main} can be used to show
that the ground state energy satisfies a derivative bound, which was used in \cite{HaslerHinrichsSiebert.2021a} to
prove the existence of a ground state for the spin boson model. This was  our main motivation to prove \cref{main}.

Adding a constant magnetic field $\mu \in \R$ to the interaction, we obtain the canonical partition function
\begin{align*}
	 \cZ_\mu(W,T)        &  = \EE\left[\exp\left(\int_{-T}^T\int_{-T}^TW(t-s)X(t)X(s)dsdt + \mu \int_{-T}^T X(t)dt\right)\right]  .
\end{align*}
 The magnetization is then defined as
$$
\mathcal{M}_\mu(W,T) = \frac{1}{T} \frac{\partial}{\partial \mu}  \ln  \cZ_\mu(W,T)
$$
and the magnetic susceptibility is defined as
$$
\mathcal{X}_\mu(W,T)  =  \frac{\partial}{\partial \mu} \mathcal{M}_\mu(W,T)  .
$$
A straightforward calculation shows that the  the magnetic susceptibility at zero satisfies
$$
 \left. \mathcal{X}_\mu(W,T) \right|_{\mu = 0 } = \frac{1}{\cZ(W,T)}
	\mathbb{E} \left[  \frac 1T  \left( \int_{-T}^T X(t) dt \right)^2  \exp\left( \int_{-T}^T \int_{-T}^T X(t) X(s) W(t-s) dt ds \right)  \right]  ,
$$
which is the expression estimated in \cref{main}.

\bigskip\noindent
Now, let us consider the spin boson model with an external magnetic field.
We sketch the  relation of  the second order derivative of the ground state energy with respect to the magnetic field
to the  expression estimated in  \cref{main}.
In \cite{HaslerHinrichsSiebert.2021a}, we proved that an upper bound on the magnetic susceptibility in the spin boson model implies existence of ground states, if it is uniform in the photon mass.
It is well-known that the ground state energy of the spin boson model can be equivalently described as a jump process, which itself is the continuum limit of the one-dimensional Ising model \cite{EmeryLuther.1974}. This duality has been used to study the spin boson model in the past \cite{FannesNachtergaele.1988,SpohnDuemcke.1985,Spohn.1989,Abdessalam.2011,HirokawaHiroshimaLorinczi.2014}. In this spirit, our result is formulated as a bound on the expectation value of a Poisson-driven jump process.

We use notation similar to \cite{HaslerHinrichsSiebert.2021a} and refer the reader to that paper for more rigorous definitions.
We fix a measurable function $\omega:\IR^d\to[0,\infty)$ with $\omega>0$ almost everywhere, an element $v\in L^2(\IR^d)$ such that $v/\sqrt \omega\in L^2(\IR^d)$ and a coupling constant $\lambda\in\IR$.
Let $\FS$ be the Fock space over $L^2(\IR^d)$ and denote by $\dG(\omega)$ and $a^*(v)$, $a(v)$  the usual second quantization operator of $\omega$ and the creation/annihilation operators corresponding to $v$, respectively.
 Further, assume $\sigma_x$ and $\sigma_z$ are the usual Pauli matrices. Then, we define the spin boson Hamiltonian with an  external field of strength $\mu\in\IR$ as the selfadjoint lower-bounded operator acting on $\IC^2\otimes\FS$ as
\begin{equation}
	H(\mu) = ( \sigma_z + \Id ) \otimes \Id + \Id \otimes \dG(\omega) + \sigma_x \otimes \left(\lambda ( a^*(v) + a(v))  + \mu \Id\right).
\end{equation}
We investigate properties of the ground state energy
\begin{equation}
	E(\mu) = \inf\sigma(H(\mu)).
\end{equation}
To that end, we use Bloch's formula. Let $\Omega_\downarrow = \begin{pmatrix}0\\1\end{pmatrix}\otimes \Omega$, where $\Omega$ denotes the Fock space vacuum.
Then Bloch's formula states that for
 \begin{equation}  \label{eq:ET}
E_T(\mu) =  - \frac 1T\ln \Braket{\Omega_\downarrow,e^{-TH(\mu)}\Omega_\downarrow}
\end{equation}
one has
\begin{equation}\label{Bloch}
	E(\mu) = \lim_{T\to\infty}  E_T(\mu) .
\end{equation}
A  rigorous  proof  of  \eqref{Bloch}  can  be obtained by a straightforward application of the  spectral theorem, provided  $1_{ H(\mu) \leq E(\mu) + \eps } \Omega_\downarrow \neq 0$ for all $\eps > 0$. This  last assumption can be shown using that $e^{ - T H(\mu)}$ is positivity improving, which has been shown for example  in  \cite{HaslerHerbst.2010}
for $\mu = 0$ and follows for arbitrary $\mu \in \R$ by a simple modification. Now the right hand side of  \cref{eq:ET}
can be calculated
using the Feynman-Kac formula and integrating out the quantum  field in the  so called Schr\"odinger  representation \cite{Simon.1979,LorincziHiroshimaBetz.2011}.
 Such a calculation yields
\begin{equation}\label{eq:feynmankac}
		 \braket{\Omega_\downarrow,e^{-TH(\mu)}\Omega_\downarrow} =  \cZ_\mu(W,T),
\end{equation}
where
\begin{equation} \label{reltoW}  W(t) = \frac {\lambda^2}8\int_{\IR^d}|v(k)|^2e^{-|t|\omega(k)}dk , \quad t \in \R . \end{equation}
We note that \cref{eq:feynmankac} has been shown in the literature for $\mu = 0$ \cite{HirokawaHiroshimaLorinczi.2014} and a similar  formula is derived in \cite{Spohn.1989} for  KMS states.  Note that the function $W(t)$ defined in  \eqref{reltoW}  is symmetric, continuous, and in $L^1(\R)$.
 Since we have not found an explicit proof of  \cref{eq:feynmankac} in the literature,  we plan to address this in a
forthcoming paper.

Now, inserting   \cref{eq:feynmankac} into    \cref{eq:ET}  we find
\begin{equation}
	E_T(\mu) =-  \frac{1}{T}\ln \cZ_\mu(W,T).
\end{equation}
Differentiating this expression  twice  with respect to $\mu$ and evaluating it  at zero,  we find from the calculation in  the first part of this section that
\begin{align*}
		E_T''(0) = - \mathcal{X}_\mu(W,T) |_{\mu=0} .
\end{align*}
Now, let us  consider  the limit $T\to\infty$. Provided one can show that
 the limit
\begin{equation} \label{limitother}
	E''(0) = \lim_{T\to\infty} E_T''(0) ,
\end{equation}
exists, one  obtains
 \begin{equation}\label{eq:seconddermu0}
	E''(0) = -\lim_{T\to\infty} \mathcal{X}_\mu(W,T) |_{\mu=0} .  
\end{equation}
Given \cref{eq:seconddermu0}, \cref{main}  yields  a bound on the second derivative of the ground state energy of the spin boson
model with respect to an external magnetic field.  This bound  is uniform  in the $L^1$-norm of $W$.
We note that     \cref{limitother} can be shown to hold for example  if the ground state energy  is isolated from the rest of the spectrum.
We plan to address this in a forthcoming paper.
It is well-known that there exists such a gap if $\inf_{k\in\IR^d}\omega(k)>0$ \cite{AraiHirokawa.1995}.

\section{ A Correlation Bound  for  the Ising Model}\label{sec:cor}

In this section, we introduce the Ising model and prove an upper bound on correlation functions, which will be  stated in \cref{prop:corbound}, below.
The novel aspect of this bound is that it can  accomodate  arbitrarily  large nearest neighbor couplings.
 This result is the main technical ingredient to our proof of \cref{main}. The connection between the jump process  and the Ising model will be treated in \cref{sec:continuumlimit}. Bounds on correlation functions of the Ising model have been studied throughout the literature, cf. \cite{Griffiths.1967a,KellySherman.1968,Ginibre.1970,Thompson.1971,RogersThompson.1981} and references therein. They are for example used to prove the existence of the thermodynamic limit and of phase transitions in the Ising model, cf. \cite{Griffiths.1967b,GallavottiMiracleSole.1967,Ruelle.1968,Dyson.1969,KacThompson.1969,FroehlichIsraelLiebSimon.1978,AizenmanNewman.1986}.

Let $L\in\IN$ and $\Lambda_L=\IZ\cap[-L,+L]$. We define the spin configuration space $\cS_L=\{-1,1\}^{\Lambda_L}$. For $\sigma=(\sigma_i)_{i\in \Lambda_L}\in\cS_L$ and $A\subset \Lambda_L$, we write
\begin{equation}
	\sigma_A = \prod_{i\in A}\sigma_i , 
\end{equation}
where we use the convention that $\sigma_\emptyset = 1$. 
For $J:\cP(\IZ)\to\IR$, we define the corresponding Ising energy
\begin{equation}
	E_{J,L}(\sigma) = -\sum_{A\subset \Lambda_L} J(A)\sigma_A
\end{equation}
and the partition function
\begin{equation}\label{def:partition}
	Z_{J,L} = \sum_{\sigma\in\cS_L}\exp(-E_{J,L}(\sigma)).
\end{equation}
In contrast  to the standard definitions in statistical mechanics, we absorb the thermodynamic parameter $\beta$ in the interaction function $J$.
The expectation value of a function $f:\cS_L\to\IR$ is now defined as
\begin{equation}\label{def:expIsing}
	\bL{f}_J = \frac{1}{Z_{J,L}}\sum_{\sigma\in\cS_L}f(\sigma)\exp(-E_{J,L}(\sigma)).
\end{equation}
For given $f:\cS_L\to\IR$ and $\tilde{L}\ge L$, we denote the function $\tilde{f}:\cS_{\tilde{L}}\to\IR$ with $\tilde{f}(\sigma) = f(\sigma|_{\Lambda_L})$ again by the same symbol $f$.
Then, if the thermodynamic limit $L\to\infty$ exists, we will drop the superscript  ${(L)}$ and write
\begin{equation}
	\braket{f}_{J} = \lim_{L\to\infty}\bL{f}_J.
\end{equation}
Especially, we note that the existence of the thermodynamic limit of correlation functions $\braket{\sigma_i\sigma_j}_J$ for $J:\cP(\IZ)\to[0,\infty)$ is well-known (cf. \cite{Griffiths.1967b} or \cref{cor:thermlim}).

For a sequence $w=(w_k)_{k\in\IN}\subset\IR$, we  define the associated pair interaction
\begin{equation}\label{def:pairinteraction}
	J_{w}: \cP(\IZ)\to\IR \quad\mbox{with}\quad \begin{cases} \{i,j\} \mapsto w_{|i-j|} & \mbox{for}\ i,j\in\IZ, i\ne j,\\ A \mapsto 0 & \mbox{for any other}\ A\subset\IZ.  \end{cases}
\end{equation}
In this section we prove the following proposition.
\begin{proposition}\label{prop:corbound}
	For  every   $\eps \in (0,\frac{1}{10})$ there exists a  $C_\eps>0$, such that for any $w=(w_k)_{k\in\IN}\in\ell^1(\IN)$ with $w\ge 0$ and \begin{equation} \label{eq:condw}
 \sum_{l=2}^{\infty}\tanh w_l  \le\eps(1-\tanh w_1) ,  \end{equation} we have
	\begin{equation} \label{eq:mainboundpropcorbound}
		\sum_{i\in\IZ}\braket{\sigma_i\sigma_j}_{J_w} \le \frac{C_\eps}{1-\tanh w_1} \qquad\mbox{for all}\ j\in\IZ.
	\end{equation}
\end{proposition}
\begin{remark}
We note that for  $v \in\ell^1(\IN)$  the  sequence  $ w = \beta v$  satisfies  the relation \eqref{eq:condw}    for sufficiently small  $\beta > 0$.
 Hence, our bound describes absence of long range order in the Ising model for any summable pair interaction provided the temperature is large  enough.
\end{remark}
\begin{remark}
	We note that  correlations estimates have been shown already a long time ago  in \cite{Dyson.1969,RogersThompson.1981}.
We generalize the result of \cite{Dyson.1969}, in the sense that  we can accomodate arbitrary large nearest neighbor couplings
and obtain an analogous correlation bound. On the other hand the assumptions in \cite{RogersThompson.1981}
or weaker but their asseration is weaker as well.
  Explicitly, Rogers and Thompson prove the estimate
	$ \lim\limits_{N\to\infty}\frac{1}{N^2}\sum_{i,j=1}^{N}\braket{\sigma_i\sigma_j}_{J_w} = 0   $
	under the assumption $\sum_{k=1}^{N}kw_k = o((\ln N)^{1/2})$, which shows the absence of long-range order.
	Note that under the stronger assumption  \cref{eq:condw}, \cref{prop:corbound} implies the  correlation estimate
	\begin{equation}
		\limsup_{N\to\infty}\frac{1}{N}\sum_{i,j=1}^{N}\braket{\sigma_i\sigma_j}_{J_w} <\infty ,
	\end{equation}
which is stronger.
\end{remark}
Let us begin with recalling some well-known inequalities on correlation functions in the Ising model, which  go back to Griffiths \cite{Griffiths.1967a,Griffiths.1967c} (later generalized in \cite{KellySherman.1968,Ginibre.1970} and referred to as the GKS (Griffiths-Kelly-Sherman) inequalities) and Thompson \cite{Thompson.1971}.
To that end, we write the symmetric set difference as $AB = A\cup B\setminus(A\cap B)$ for $A,B\subset \IZ$.
Further, if $\cA\subset\cP(\IZ)$, we
 define
\begin{equation}\label{def:cutising}
	\bL{\cdot}_{J;\cA} := \bL{\cdot}_{I_\cA} \quad\mbox{and}\quad \braket{\cdot}_{J;\cA} := \braket{\cdot}_{I_\cA}, \qquad\mbox{where}\  I_\cA(A) = \begin{cases}J(A) & \mbox{for}\ A\notin\cA,\\ 0 & \mbox{for}\ A\in \cA.\end{cases}
\end{equation}
\begin{lemma}\label{lem:standardbounds}
	Let $J:\cP(\IZ)\to[0,\infty)$ and assume $A,B\subset\IZ$, $L\in\IN$. Then the following holds.
	\begin{enumerate}[(i)]
		\item\label{part:G1} $\bL{\sigma_A}_{J}\ge 0$ {\em (Griffiths' first inequality)}
		\item\label{part:G2} $\bL{\sigma_{AB}}_{J} \ge \bL{\sigma_A}_{J}\bL{\sigma_B}_{J}$ {\em (Griffiths' second inequality)}
		\item\label{part:T1} $\bL{\sigma_A}_{J} \le \bL{\sigma_A}_{J;\{B\}}+\tanh(J(B))\bL{\sigma_{AB}}_{J;\{B\}}$ {\em (Griffiths' third inequality)}
	\item \label{part:simple} $\bL{\sigma_A}_{J;\{B\}}\le \bL{\sigma_A}_{J}$
		\item\label{part:T2} $\bL{\sigma_A}_{J} \le \tanh (J(B))\bL{\sigma_{AB}}_{J} + (1-\tanh^2(J(B)))\bL{\sigma_A}_{J;\{B\}}$
	\end{enumerate}
\end{lemma}
\begin{proof} Parts
	\cref{part:G1,part:G2} follow from the main theorem in \cite{KellySherman.1968}.
	Parts \cref{part:T1}--\cref{part:T2} are shown in  \cite{Thompson.1971} in (3.1), (1.6),  and (2.5), respectively.
\end{proof}
\noindent
We will also utilize the  well-established simple fact  that expectations involving uncoupled Ising spins always vanish. This is the content of the next lemma.
\begin{lemma}\label{lem:zeroexp}
	Let $L\in\IN$, $i\in\Lambda_L$ and assume $J:\cP(\IZ)\to\IR$ satisfies $J(A)=0$ for all $A\subset\Lambda_L$ with $i \in A$.
	Then   $\bL{\sigma_B}_{J}=0$  for any $B \subset\Lambda_L$ with $i\in B$.
\end{lemma}
\begin{proof}
	We define $\phi_i:\cS_L\to\cS_L$ as $(\phi_i(\sigma))_k = - \sigma_k$, if $k=i$, and $(\phi_i(\sigma))_k =  \sigma_k$, if $k \neq i$.
	By the assumptions, it follows that $E_{J,L}(\phi_i(\sigma))=E_{J,L}(\sigma)$ for all $\sigma\in\cS_L$. Further,
	if $i\in B$, we have $\sigma_B \circ \phi_i =-\sigma_B$. Together, we obtain
	\[ \bL{\sigma_B}_{J}=\bL{\sigma_B \circ \phi_i}_{J}=\bL{-{\sigma_B}}_{J} = -\bL{\sigma_B}_{J}. \]
This implies the  claim.
\end{proof}
\noindent
The existence of the thermodynamic limit immediately follows from \cref{lem:standardbounds} and is well-known since \cite{Griffiths.1967b}.
\begin{cor}\label{cor:thermlim}
		Let $J:\cP(\IZ)\to[0,\infty)$ and assume $A\subset \IZ$.
		Then the thermodynamic limit $\braket{\sigma_A}_{J}$ exists.
\end{cor}
\begin{proof}
By \cref{lem:standardbounds}, the expectation  $\bL{\sigma_A}_{J}$ is nonnegative (Part \eqref{part:G1}), increasing in $L$  (Part  \eqref{part:simple}), and bounded above by 1. Thus  the statement follows by monotone convergence.
\end{proof}
\noindent
The major ingredient of the proof of \cref{prop:corbound} is the following correlation bound for  finite Ising spin chains.
\begin{lemma}\label{lem:isingboundsfixedL}
	Let $L\in\IN$ and $w=(w_k)_{k\in\IN}\subset[0,\infty)$. We set $\tau_k = \tanh(w_k)$.\\
	If $i,j\in\Lambda_L$ with $i \lessgtr j$, we have
	\[
	\bL{\sigma_i\sigma_j}_{J_w}
	\le \tau_1\bL{\sigma_i\sigma_{j \mp 1}}_{J_w} + \sum_{l=2}^{\infty}\sum_{s=\pm 1}\tau_{l}\bL{\sigma_i\sigma_{j+sl}}_{J_w} + (1-\tau_1^2)\sum_{b=1}^{\infty}\tau_1^b\sum_{l=2}^{\infty}\sum_{s=\pm 1} \tau_{l}\bL{\sigma_i\sigma_{j\pm b+sl}}_{J_w} ,
	\]
where we use the convention that $\bL{\sigma_l\sigma_{k}}_J = 0$ if $l $ or $k$ is not an element of $\Lambda_L$.
\end{lemma}

\begin{proof}
	The philosophy of our proof is sketched in \cref{fig}.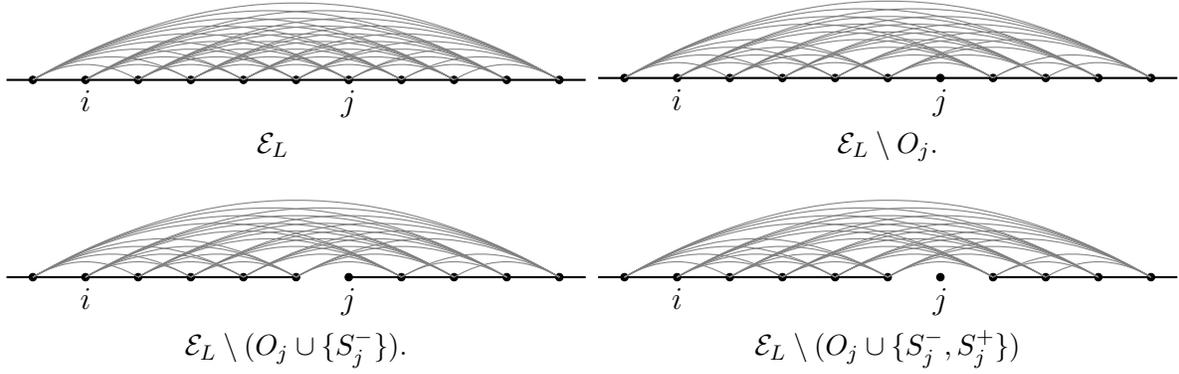
\begin{figure}
	\begin{center}
		\begin{tikzpicture}[scale=.7]
			\draw[black, thick] (-0.5,0) -- (10.5,0);
			\foreach \x in {0, 1,2,3,4,5,6,7,8,9,10} 
			\filldraw[black] (\x,0) circle (2pt); 
			\draw (1,0)  node[anchor=north] {$i$};
			\draw (6,0)  node[anchor=north] {$j$};
			\foreach \x in {0, 1,2,3,4,5,6,7,8} 
			\draw[gray] (\x,0) to[bend left] (\x+2,0);
			\foreach \x in {0, 1,2,3,4,5,6,7} 
			\draw[gray] (\x,0) to[bend left] (\x+3,0);
			\foreach \x in {0, 1,2,3,4,5,6} 
			\draw[gray] (\x,0) to[bend left] (\x+4,0);
			\foreach \x in {0, 1,2,3,4,5} 
			\draw[gray] (\x,0) to[bend left] (\x+5,0);
			\foreach \x in {0, 1,2,3,4} 
			\draw[gray] (\x,0) to[bend left] (\x+6,0);
			\foreach \x in {0, 1,2,3} 
			\draw[gray] (\x,0) to[bend left] (\x+7,0);
			\foreach \x in {0, 1,2} 
			\draw[gray] (\x,0) to[bend left] (\x+8,0);
			\foreach \x in {0, 1} 
			\draw[gray] (\x,0) to[bend left] (\x+9,0);
			\foreach \x in {0} 
			\draw[gray] (\x,0) to[bend left] (\x+10,0);
			\draw (5,-0.8)  node[anchor=north] {\small $\mathcal{E}_L   \phantom{;O_j}$};
		\end{tikzpicture}
		\begin{tikzpicture}[scale=.7]
			\draw[black, thick] (-0.5,0) -- (10.5,0);
			\foreach \x in {0, 1,2,3,4,5,6,7,8,9,10} 
			\filldraw[black] (\x,0) circle (2pt); 
			\draw (1,0)  node[anchor=north] {$i$};
			\draw (6,0)  node[anchor=north] {$j$};
			\foreach \x in {0, 1,2,3,5,7,8} 
			\draw[gray] (\x,0) to[bend left] (\x+2,0);
			\foreach \x in {0, 1,2,4,5,7} 
			\draw[gray] (\x,0) to[bend left] (\x+3,0);
			\foreach \x in {0, 1,3,4,5} 
			\draw[gray] (\x,0) to[bend left] (\x+4,0);
			\foreach \x in {0,2,3,4,5} 
			\draw[gray] (\x,0) to[bend left] (\x+5,0);
			\foreach \x in { 1,2,3,4} 
			\draw[gray] (\x,0) to[bend left] (\x+6,0);
			\foreach \x in {0, 1,2,3} 
			\draw[gray] (\x,0) to[bend left] (\x+7,0);
			\foreach \x in {0, 1,2} 
			\draw[gray] (\x,0) to[bend left] (\x+8,0);
			\foreach \x in {0, 1} 
			\draw[gray] (\x,0) to[bend left] (\x+9,0);
			\foreach \x in {0} 
			\draw[gray] (\x,0) to[bend left] (\x+10,0);
			\draw (5,-0.8)  node[anchor=north] {\small   $\mathcal{E}_L \setminus O_j$.};
		\end{tikzpicture}
	\begin{tikzpicture}[scale=.7]
		\draw[black, thick] (-0.5,0) -- (5,0);
		\draw[black, thick] (6,0) -- (10.5,0);
		\foreach \x in {0, 1,2,3,4,5,6,7,8,9,10} 
		\filldraw[black] (\x,0) circle (2pt); 
		\draw (1,0)  node[anchor=north] {$i$};
		\draw (6,0)  node[anchor=north] {$j$};
		\foreach \x in {0, 1,2,3,5,7,8} 
		\draw[gray] (\x,0) to[bend left] (\x+2,0);
		\foreach \x in {0, 1,2,4,5,7} 
		\draw[gray] (\x,0) to[bend left] (\x+3,0);
		\foreach \x in {0, 1,3,4,5} 
		\draw[gray] (\x,0) to[bend left] (\x+4,0);
		\foreach \x in {0,2,3,4,5} 
		\draw[gray] (\x,0) to[bend left] (\x+5,0);
		\foreach \x in { 1,2,3,4} 
		\draw[gray] (\x,0) to[bend left] (\x+6,0);
		\foreach \x in {0, 1,2,3} 
		\draw[gray] (\x,0) to[bend left] (\x+7,0);
		\foreach \x in {0, 1,2} 
		\draw[gray] (\x,0) to[bend left] (\x+8,0);
		\foreach \x in {0, 1} 
		\draw[gray] (\x,0) to[bend left] (\x+9,0);
		\foreach \x in {0} 
		\draw[gray] (\x,0) to[bend left] (\x+10,0);
		\draw (5,-0.8)  node[anchor=north] {\small     $\mathcal{E}_L \setminus (  O_j \cup \{ S_j^{-} \} ) $.};
	\end{tikzpicture}
	\begin{tikzpicture}[scale=.7]
		\draw[black, thick] (-0.5,0) -- (5,0);
		\draw[black, thick] (7,0) -- (10.5,0);
		\foreach \x in {0, 1,2,3,4,5,6,7,8,9,10} 
		\filldraw[black] (\x,0) circle (2pt); 
		\draw (1,0)  node[anchor=north] {$i$};
		\draw (6,0)  node[anchor=north] {$j$};
		\foreach \x in {0, 1,2,3,5,7,8} 
		\draw[gray] (\x,0) to[bend left] (\x+2,0);
		\foreach \x in {0, 1,2,4,5,7} 
		\draw[gray] (\x,0) to[bend left] (\x+3,0);
		\foreach \x in {0, 1,3,4,5} 
		\draw[gray] (\x,0) to[bend left] (\x+4,0);
		\foreach \x in {0,2,3,4,5} 
		\draw[gray] (\x,0) to[bend left] (\x+5,0);
		\foreach \x in { 1,2,3,4} 
		\draw[gray] (\x,0) to[bend left] (\x+6,0);
		\foreach \x in {0, 1,2,3} 
		\draw[gray] (\x,0) to[bend left] (\x+7,0);
		\foreach \x in {0, 1,2} 
		\draw[gray] (\x,0) to[bend left] (\x+8,0);
		\foreach \x in {0, 1} 
		\draw[gray] (\x,0) to[bend left] (\x+9,0);
		\foreach \x in {0} 
		\draw[gray] (\x,0) to[bend left] (\x+10,0);
		\draw (5,-0.8)  node[anchor=north] {\small    $\mathcal{E}_L \setminus ( O_j \cup \{ S_j^-,S_j^+\} )$     };
	\end{tikzpicture}
	\end{center}
	\caption{Illustration of the set $\mathcal{E}_L$, consisting of all  edges with vertices in $\Lambda_L$,  without the edges of the indicated sets.}\label{fig}
\end{figure}
	We use the estimates in \cref{lem:standardbounds} to reduce the number of interaction edges, in which $j$ contributes.
	To that end, for $j\in\Lambda_L$, we define the sets
	\[ S_j^\pm = \{j,j\pm1\} \quad \mbox{and}\quad O_j = \{\{j,k\}:k\in\IZ\setminus\{j,j-1,j+1\}\}.  \]
	Note that $S_j^\pm$ contain the nearest neighbors of $j$, while $O_j$ are all long-range pairs involving $j$.
	Throughout this proof, we drop the superscript ${(L)}$ and the subscript ${J_w}$ of expectation values. Moreover we  assume $i<j$. The statement in the case $i>j$ can be treated completely analogous.
	
	By twice applying \cref{lem:standardbounds} \cref{part:T1}, we obtain
	\begin{align*}
		\braket{\sigma_i\sigma_j} &\le \braket{\sigma_i\sigma_j}_{;\{\{j,j-2\}\}} + \tau_2 \braket{\sigma_i\sigma_{j-2}}_{;\{\{j,j-2\}\}} \\&\le \braket{\sigma_i\sigma_j}_{;\{\{j,j-2\},\{j,j+2\}\}}
+ \tau_2 \braket{\sigma_i\sigma_{j+2}}_{;\{\{j,j-2\},\{j,j+2\}\}}
+ \tau_2 \braket{\sigma_i\sigma_{j-2}}_{;\{\{j,j-2\}\}} .
	\end{align*}
	Combined with \cref{lem:standardbounds} \cref{part:simple}, this implies
	\[ \braket{\sigma_i\sigma_j} \le \tau_2\left(\braket{\sigma_i\sigma_{j-2}}+\braket{\sigma_i\sigma_{j+2}}\right) + \braket{\sigma_i\sigma_j}_{;\{\{j,j-2\},\{j,j+2)\}\}}. \]
	Iterating this argument, we arrive at
	\begin{equation}\label{eq:firstiteration}
		\braket{\sigma_i\sigma_j} \le \sum_{l=2}^{\infty}\tau_l \sum_{s=\pm}\braket{\sigma_i\sigma_{j+sl}} + \braket{\sigma_i\sigma_j}_{;O_j}.
	\end{equation}
	Then, \cref{lem:standardbounds} \cref{part:T2} yields
	\begin{equation}\label{eq:stepafterfirstit}
		\braket{\sigma_i\sigma_j}_{;O_j}\le \tau_1\braket{\sigma_i\sigma_{j-1}}_{;O_j}+(1-\tau_1^2)\braket{\sigma_i\sigma_j}_{;O_j \cup \{S_j^-\}}.
	\end{equation}
	The second term on the right hand side can be estimated by \cref{lem:standardbounds} \cref{part:T1} and \cref{lem:zeroexp}
	\begin{equation}\label{eq:firstzero}
		\braket{\sigma_i\sigma_j}_{;O_j \cup \{ S_j^-\}} \le \underbrace{\braket{\sigma_i\sigma_j}_{;O_j \cup \{ S_j^-,S_j^+\}}}_{=0} + \tau_1\braket{\sigma_i\sigma_{j+1}}_{;O_j \cup \{S_j^-,S_j^+\}}.
	\end{equation}
	Now applying \cref{eq:firstiteration} with $j$ replaced by $j+1$ and using \cref{lem:standardbounds} \cref{part:simple}, we obtain
	\begin{equation}\label{eq:inserthere}
		\braket{\sigma_i\sigma_{j+1}}_{;O_j \cup \{ S_j^-,S_j^+\}} \le \sum_{l=2}^\infty\sum_{s=\pm 1 }\tau_{l}\braket{\sigma_i\sigma_{j+1+sl}} + \braket{\sigma_i\sigma_{j+1}}_{;O_j  \cup  O_{j+1}\cup \{ S_j^-,S_j^+\}}.
	\end{equation}
	As in \cref{eq:firstzero}, we use \cref{lem:standardbounds} \cref{part:T1} and \cref{lem:zeroexp}, which yield
	\begin{equation}\label{eq:itlast}
		\braket{\sigma_i\sigma_{j+1}}_{;O_j \cup O_{j+1} \cup \{ S_j^-,S_j^+\}} \le \underbrace{\braket{\sigma_i\sigma_{j+1}}_{;O_j \cup O_{j+1} \cup \{  S_j^-,S_j^+, S_{j+1}^+\}}}_{=0} + \tau_1\braket{\sigma_i\sigma_{j+1}}_{;O_j\cup O_{j+1}  \cup \{ S_j^-,S_j^+ ,  S_{j+1}^+\}}.
	\end{equation}
	Note that we hereby used $S_j^+=S_{j+1}^-$.
	We now insert \cref{eq:itlast} into \cref{eq:inserthere} and iterate the same arguments. As a result
	\begin{equation}\label{eq:seconditerateresult}
		\braket{\sigma_i\sigma_{j+1}}_{;O_j \cup \{S_j^-,S_j^+\}} \le \sum_{b=1}^{\infty} \sum_{l=2}^\infty \sum_{s=\pm 1}\tau_1^{b-1}\tau_l\braket{\sigma_i\sigma_{j+b+sl}}.
	\end{equation}
	The statement now follows by combining \cref{eq:firstiteration,,eq:stepafterfirstit,,eq:firstzero,,eq:seconditerateresult}.
\end{proof}
\noindent
We use the previous lemma to prove the central result of this section.
\begin{proof}[\textbf{Proof of \cref{prop:corbound}}]
	For the proof of the statement, we will use the estimate from \namecref{lem:isingboundsfixedL} \labelcref{lem:isingboundsfixedL}.
We  need  to take  the limit $L \to \infty$ and  sum  over all $i\in\IZ$ .
To show finiteness we will make use of translation invariance of the model.
Let us first assume that  $w\in\ell^1(\IN)$ with   $w   \geq 0$ has compact support and let  $K > 0$ be such that \begin{equation} \label{eq:comp}
w_k = 0   , \quad k \geq K .
\end{equation}
As in  \cref{lem:isingboundsfixedL} we shall use the  notation $\tau_k = \tanh(w_k)$.
 We introduce a regularization parameter $\eta>0$ and define
	\begin{equation}\label{eq:reg}
		\tau_{k,\eta}=e^{\eta k}\tau_k \quad \mbox{and}\quad \bLe{\sigma_i\sigma_j}_{J_w}=e^{-\eta|i-j|}\bL{\sigma_i\sigma_j}_{J_w}.
	\end{equation}
	Further, we define
	\begin{align*}
		&M_{j,L}^-(\eta) = \sum_{i=-L}^{j-1}\bLe{\sigma_i\sigma_j}_{J_w},
		\quad
		M_{j,L}^+(\eta) = \sum_{i=j+1}^{L}\bLe{\sigma_i\sigma_j}_{J_w},
		\quad\mbox{and}
		\\
		& M_{j,L}(\eta) = \sum_{i=-L}^{L}\bLe{\sigma_i\sigma_j}_{J_w} = 1+ M_{j,L}^+(\eta)+M_{j,L}^-(\eta).
	\end{align*}
	By the regularization \cref{eq:reg} and \cref{cor:thermlim}, the limits
	$$ M_j^\pm(\eta) = \lim_{L\to\infty}M_{j,L}^\pm(\eta) = \sum_{i\gtrless j}\bLe{\sigma_i\sigma_j}_{J_w} \quad\mbox{and}\quad M_j(\eta) = \lim_{L\to\infty}M_{j,L}(\eta) = \sum_{i\in\IZ}\bLe{\sigma_i\sigma_j}_{J_w} $$
	exist for any $\eta>0$. By  translation invariance of $J_{w}$, i.e., $\braket{\sigma_i\sigma_j}_{J_{w}}=\braket{\sigma_{i+k}\sigma_{j+k}}_{J_{w}}$ for any $k\in\IZ$, it
follows that  $M_{j}(\eta)$ and $M_{j}^\pm(\eta)$ are independent of $j$ and we shall write  $M(\eta)$ for $M_j(\eta)$.

	For $L\in\IN$, we now multiply the inequalities in \cref{lem:isingboundsfixedL} with $e^{-\eta|i-j|}$ and use the triangle inequality, to obtain for  $ i \lessgtr  j$
	\begin{align*}
		 \bLe{\sigma_i\sigma_j}_{J_w}
		\le & \tau_{1}\bLe{\sigma_i\sigma_{j\mp1}}_{J_w} + \sum_{l=2}^{\infty}\sum_{s=\pm 1}\tau_{l,\eta}\bLe{\sigma_i\sigma_{j+sl}}_{J_w} \\&\qquad+ (1-\tau_1^2)\sum_{b=1}^{\infty}\tau_{1,\eta}^b\sum_{l=2}^{\infty}\sum_{s=\pm 1} \tau_{l,\eta}\bLe{\sigma_i\sigma_{j\pm b+sl}}_{J_w} .
	\end{align*}
	Adding the above expression for the cases $i > j$ and $j > i$, summing  over all $i\in\Lambda_L$, using $\sigma_r^2  = 1$ for any $r \in \Z$  as well as \cref{lem:standardbounds} \cref{part:G1}, we find
	\begin{equation}\label{eq:boundforfiniteL}
		\begin{aligned}
			M_{j,L}(\eta)\le &  1+ \tau_{1}\left( M_{j-1,L}^-(\eta) + 2 + M_{j+1,L}^+(\eta) \right)
			 + \sum_{l=2}^{K}\tau_{l,\eta}\sum_{s=\pm 1}M_{j+sl,L}(\eta) \\
			& + \sum_{b=1}^{\infty}\tau_{1,\eta}^b(1-\tau_1^2)\sum_{l=2}^{K}\tau_{l,\eta}\sum_{s=\pm 1}\left(M_{j+b+sl,L}(\eta)+M_{j-b+sl,L}(\eta)\right).
		\end{aligned}
	\end{equation}
Now we can take the limit $L \to \infty$. Since $\tau$ has  compact support and $\eta > 0$, expressions on the right hand side stay finite.
	Then, using the translation invariance of $J_{w}$ we can   drop the index $j$, and summing the geometric series $\sum_{b\in\IN}\tau_{1,\eta}^b$, we obtain
	\begin{equation}\label{eq:MKfinish}
		M(\eta) \le 1 + \tau_1 + M(\eta) \left(\tau_1 +  2 \sum_{l=2}^{K}\tau_{l,\eta}\left(1+2\frac{1-\tau_1^2}{1-\tau_{1,\eta}}\right)\right).
	\end{equation}
	Fix $D>1$, such that $\eps\in(0,(10D)^{-1})$.
	Since
	\begin{equation}\label{eq:simplebound}
		1 < \frac{1-\tau_1^2}{1-\tau_1}=1+\tau_1 <2,
	\end{equation}
	we can choose $\eta_0>0$, such that 
	$\frac{1-\tau_1^2}{1-\tau_{1,\eta_0}}<2$ and $e^{K\eta_0}<D$. 
	Then, for any $\eta\in(0,\eta_0)$  we obtain
	\begin{equation}\label{eq:tausum}
		\sum_{l=2}^{K}\tau_{l,\eta} \le D\sum_{l=2}^{K}\tau_l 
	\end{equation}
	and
	\begin{equation}\label{eq:prefactor}
		\tau_{1} + 2 \sum_{l=2}^{K}\tau_{l,\eta}\left(1+2\frac{1-\tau_1^2}{1-\tau_{1,\eta}}\right) < \tau_1 + 10 D  \sum_{l=2}^{\infty}\tau_l  .
	\end{equation}
	If
\begin{equation}  \label{eq:ineqmainassump}
\sum_{l=2}^{\infty}\tau_l   \le \eps (1-\tau_1) ,
\end{equation}
the right hand side of
	\cref{eq:prefactor} is smaller than 1, and  we can bring $M(\eta)$ in  \cref{eq:MKfinish}  to the left hand
side. Thus using \cref{eq:simplebound,eq:prefactor}, we find
	\[  M (\eta) \le \frac{1+\tau_1}{\displaystyle 1-\tau_1 - 2  \sum_{l=2}^{K}\tau_{l,\eta}\left(1+2\frac{1-\tau_1^2}{1-\tau_{1,\eta}}\right)} \le
	\frac{2}{\displaystyle 1-\tau_1- 10   D   \sum_{l=2}^{\infty}\tau_l   } \le
	\frac{2}{1-  10 D\eps}\frac{1}{1-\tau_1}.\]
	By monotone convergence, the limit $\eta\downarrow0$ exists and
	\begin{equation}  \sum_{i\in \IZ}\braket{\sigma_i\sigma_j}_{J_{w}} = \lim_{\eta\downarrow 0} M(\eta) \le  \frac{2}{1-  10 D\eps}\frac{1}{1-\tau_1}. \label{eq:mainineqprop41}  \end{equation}
Thus, we have proven \cref{eq:mainineqprop41}  for all nonnegative  $w \in \ell^1(\IN)$  satisfying  \cref{eq:comp,eq:ineqmainassump}.

Finally, let us consider general   $w\in\ell^1(\IN)$ with   $w   \geq 0$ satisfying only \cref{eq:ineqmainassump}.  If  $i, j \in \Lambda_L$, then as an immediate consequence
of the definition \cref{def:expIsing}
$$
 \braket{\sigma_i\sigma_j}_{J_{w}}^{(L)} =
\braket{\sigma_i\sigma_j}_{J_{w 1_{[0,2L+1]}}}^{(L)} .
$$
Since $w 1_{[0,2L+1]}$ trivially satisfies \eqref{eq:ineqmainassump} because $w$ does,
 we find from \eqref{eq:mainineqprop41} and monotonicity (\cref{lem:standardbounds} \cref{part:T1} ) the estimate for all $N \in \N$ 
$$
\sum_{i = - N}^N \braket{\sigma_i\sigma_j}_{J_{w}}^{(L)} 
 \leq
	\frac{2}{1-  10  D\eps}\frac{1}{1-\tau_1}.
$$
Thus the bound \eqref{eq:mainboundpropcorbound}  of the proposition now follows by taking in the above  inequality  first the limit $L\to \infty$ and then $N \to \infty$.
\end{proof}

\section{The Continuum Limit of the Ising Model}\label{sec:continuumlimit}

In this section we prove that the jump process $X$ defined in \cref{def:X} is the continuum limit of a one-dimensional Ising model defined as in the previous \lcnamecref{sec:cor}. The approach we use is based on the description in \cite{SpohnDuemcke.1985,Spohn.1989}.
To that end, we use a parameter $\delta\in(0,\infty)$ as lattice spacing of the discrete Ising model and define the map
\begin{equation}\label{def:id}
	\id : \IR\to \IN, \qquad t\mapsto \left\lfloor\frac{t}{\delta}+\frac 12\right\rfloor,
\end{equation}
where $\lfloor\cdot\rfloor$ as usually denotes the integer part. Note, the interval $[-T,T]$ is mapped to the lattice $\Lambda_{\Ld(T)}$ with $\Ld(T)=\id(T)$.
We set the nearest neighbor interaction on this lattice to be
\begin{equation}\label{def:jd}
	\jd = - \frac 12\ln(\delta).
\end{equation}
For a function $W:\IR\to\IR$, we define the corresponding pair interaction (cf. \cref{def:pairinteraction}) on the lattice as $\wde=(\wde_k)_{k\in\IN}$ with
\begin{equation}\label{def:wd}
	\wde_k = \delta^2 W(\delta k).
\end{equation}
We define the expectation values in the Ising model given with these interactions as
\begin{equation}
	\bn{\cdot}_{\delta,T} := \braket{\cdot}^{(\Ld(T))}_{J_{(\jd,0,\ldots)}    }
	\qquad\mbox{and}\qquad
	\bw{\cdot}_{\delta,T} := \braket{\cdot}^{(\Ld(T))}_{J_{(\jd,0,\ldots)  +\wde}}.
\end{equation}
In this section we prove the following proposition.
\begin{proposition}\label{prop:isinglimit}
	Assume $W:\IR\to\IR$ is even and continuous, $T>0$ and
	 $-T\le t_1\le \cdots\le t_N\le T$. Then
	\begin{equation*}
		\begin{aligned}
			\lim_{\delta\downarrow0}&\bw{\sd{t_1}\cdots \sd{t_N}}_{\delta,T}
			= \frac{1}{\cZ(W,T)}\EE\left[X(t_1)\cdots X(t_N)\exp\left(\int_{-T}^T\int_{-T}^T W(t-s)X(t)X(s) dt ds\right)\right].
		\end{aligned}
	\end{equation*}
\end{proposition}
\noindent
As a first step of our proof, we recall the following lemma.
\begin{lemma}\label{corrjump}  Let $t_1 \le \cdots \le t_N$ be an increasing sequence of times. Then, we have
	\[\EE[X(t_1)\cdots X(t_{N})] = e^{-2(|t_2-t_1|+\cdots+|t_{N}-t_{N-1}|)} \qquad\mbox{if}\ N\ \mbox{is even}\]
	and $\EE[X(t_1)\cdots X(t_N)]=0$ if $N$ is odd.
\end{lemma}
\begin{proof} For a simple proof see for example \cite[Lemma 1]{Abdessalam.2011}.
\end{proof}
\noindent
It is well-known, that the expectation values of Ising models only with nearest neighbor coupling can be calculated explicitly, see also \cref{app:isingnene}. In the limit $\delta\to0$, we use this to obtain the jump process $X$.
\begin{lemma}\label{lem:isinglimitnenesimple}
	Let $-T \le t_1 \le \cdots \le t_N\le T$ be an increasing sequence of times. Then
	\[ \lim_{\delta\downarrow0}\bn{\sd{t_1}\cdots\sd{t_N}}_{\delta,T} = \EE[X(t_1)\cdots X(t_N)]. \]
\end{lemma}
\begin{proof} If $N$ is odd both sides vanish (\cref{lem:isinglimitnenesimple,lemapp:isingnene}), so we assume $N$ is even.
	Then, the definition \cref{def:id,lemapp:isingnene} \cref{part:nenecor} yield
	\[ \bn{\sd{t_1}\cdots \sd{t_N}}_{\delta,T} = (\tanh \jd)^{|\id(t_2)-\id(t_1)|+\cdots+|\id(t_N)-\id(t_{N-1})|}.  \]
	Since \cref{def:id} also yields $\frac{|u-v|}{\delta}-1\le|\id(u)-\id(v)|\le \frac{|u-v|}{\delta}+1$ for all $u,v\in\IR$, we obtain
	\[  \left[(\tanh \jd)^{\delta^{-1}}\right]^{|t_2-t_1|+\cdots+|t_N-t_{N-1}|-\delta N} \le  \bn{\sd{t_1}\cdots \sd{t_N}}_{\delta,T} \le \left[(\tanh \jd)^{\delta^{-1}}\right]^{|t_2-t_1|+\cdots+|t_N-t_{N-1}|+\delta N}. \]
	Using $\lim\limits_{\delta\downarrow0}(\tanh \jd)^{\delta^{-1}} = e^{-2}$, the statement follows by \cref{corrjump}.
\end{proof}
\noindent
\cref{lem:isinglimitnenesimple}  shows  \cref{prop:isinglimit} in the  case  $W=0$. To show the proposition
for nonzero $W$, we will use the notion of weak convergence of measures, as outlined in   \cite{SpohnDuemcke.1985,Spohn.1989}.
To this end, we introduce the following definitions and recall elementary properties,
which can be found in \cite[Chapter 3]{Billingsley.1999}. We define $\cD_T$ to be the set of all right-continuous functions $\omega:[-T,T]\to\{\pm 1\}$ with finitely many jumps.
We  equip $\cD_T$ with the so-called Skorokhod topology. That is, if $\Phi_T$ denotes the set of all continuous strictly increasing bijections $\ph:[-T,T]\to[-T,T]$, we define the metric
\begin{equation}\label{def:skorokhod}
	d(\omega,\nu) = \inf_{\ph\in\Phi_T}\left(\|\ph-\Id\|_\infty + \|\omega - \nu\circ \ph\|_\infty\right)
	\qquad\mbox{for}\ \omega,\nu\in\cD_T.
\end{equation}
The topology induced on  $\cD_T$  by $d$ is the Skorokhod topology. We equip $\cD_T$ with the Borel $\sigma$-algebra.
There exists a probability  measure $P_X$ on $\cD_T$, such that for  $\omega \in \cD_T$
the jump process is given by   $X(t)(\omega) = \omega(t)$ for  $t \in [-T,T]$
and for  any  measurable function $f :     \cD_T \to \IR$
\[ \EE[f(X|_{[-T,T]})] = \int f(\omega)dP_X(\omega).  \]
In the following two lemmas,  we will assume this realization of the  jump process $X$. 
We define  $\fd: \cS_{\Ld(T)}\to\cD_T$ by  $\fd(\sigma) = \left[t \mapsto \sd{t}\right]$.
\begin{lemma}\label{lem:isingweakmeasure}
	Let $f:\cD_T\to\IR$ be bounded and continuous, $N\in\IN_0$ and $-T\le t_1\le \cdots\le t_N \le T$. Then	\[\lim_{\delta\downarrow0}\Bn{\sd{t_1}\cdots\sd{t_N}f(\fd(\sigma))}_{\delta,T} = \EE[X(t_1)\cdots X(t_N)f(X)].\]
\end{lemma}
\begin{remark}
	In fact, we prove the stronger statement $\Bn{f(\fd(\sigma))}_{\delta,T}\xrightarrow{\delta\downarrow0}\EE[f(X)]$ for any bounded  measurable function $f:\cD_T\to\IR$  for which the set of discontinuities $U_f$ satisfies $P_X(U_f) =0$.
\end{remark}
\begin{remark} The proof of \cref{lem:isingweakmeasure} is based on
weak convergence of measures. To obtain weak convergence, we will show
tightness of the associated probability measures by a combinatorical estimate.
We note that tightness can in fact been  shown
by reflection positivity    \cite{Spohn.1989}.
\end{remark} 
\begin{proof}
	We prove that the measures on $\cD_T$ associated to the nearest neighbor Ising model weakly converge to the measure given by the jump process $X$. Then, the statement follows by the Portmanteau theorem \cite[Theorem 3.16]{Klenke.2020}. To prove weak convergence, we need to combine the convergence of moments from \cref{lem:isinglimitnenesimple} and the tightness of the Ising measures, cf. \cite[Theorem 13.1]{Billingsley.1999}.
	
	 For $\delta>0$, let $P_\delta$ be the pushforward measure on $\cD_T$ obtained from the Ising probability measure on $\cS_{\Ld(T)}$ through the (obviously measurable) map $\fd$, i.e.,
	\begin{equation}\label{eq:isingprob}
		P_\delta(A) = \sum_{\sigma\in \fd^{-1}(A)} \frac{e^{-E_{J_\jd,\Ld(T)}(\sigma)}}{Z_{J_\jd,\Ld(T)}} \qquad\mbox{for all measurable sets}\ A\subset \cD_T.
	\end{equation}
	Hence,
	\[ \Bn{\sd{t_1}\cdots\sd{t_N}f(\fd(\sigma))}_{\delta,T} = \int \omega(t_1)\cdots \omega(t_N)f(\omega)dP_\delta(\omega).\]
	Now, for any $k\in\IN$ and $t=(t_1,\ldots,t_k)\in[-T,T]^k$, we define the projections $\pi_t:\cD_T\to\{\pm 1\}^k,\omega\mapsto (\omega(t_1),\ldots,\omega(t_k))$. Observe that the expectation values in \cref{lem:isinglimitnenesimple} uniquely determine the probability measures $P_\delta\circ \pi_t^{-1}$ and $P_X\circ\pi_t^{-1}$, respectively, since every function on the set $\{-1,1\}$ is given as a linear combination of the constant function  one and  the identity function.
	 Hence, \cref{lem:isinglimitnenesimple} implies the weak convergence of $P_\delta\circ\pi_t^{-1}$ to $P_X\circ \pi_t^{-1}$.
	 To deduce weak convergence of $P_\delta$ to $P_X$ as $\delta\downarrow0$, we need to prove that the family $(P_\delta)$ is tight (cf. \cite[Theorem 13.1]{Billingsley.1999}).
	 Let us reformulate this statement similar to \cite[Theorem 13.2]{Billingsley.1999}.
	 For $\eps>0$, we denote by $\Omega_\eps$ the set of all $\omega\in\cD_T$ having two discontinuities with a distance less than $\eps$, i.e.,
	 \[ \Omega_\eps = \left\{\omega\in\cD_T:\exists t_1,t_2\in(-T,T):|t_2-t_1|<\eps,\lim_{t\uparrow t_1} \omega(t)\ne \omega(t_1), \lim_{t\uparrow t_2} \omega(t) \ne \omega(t_2)\right\}.  \]
	 The family $(P_\delta)$ is tight if and only if
	 \begin{equation}\label{eq:tightness}
	 	\lim_{\eps\downarrow 0} \limsup_{\delta\downarrow 0} P_\delta(\Omega_\eps) = 0.
	 \end{equation}
 	For now fix $\eps>0$ and $\delta\in(0,\eps)$. For $\sigma\in\cS_{\Ld(T)}$, we denote by $n_\sigma$ the number of sign changes (cf. \cref{lemapp:isingnene} \cref{part:neneen}). We observe that $\fd(\sigma)\in\Omega_\eps$ if $n_\sigma>2T/\eps$. Otherwise, $\fd(\sigma)\notin\Omega_\eps$ if and only if all sign changes have a distance of at least $\eps/\delta$. If $n_\sigma=k$ for some fixed $k\in\IN$, then simple combinatorics yield that there are $\binom{2\Ld(T)-(k-1)\lfloor\eps/\delta\rfloor}{k}$ possibilities to position the sign changes, such that all distances are larger than $\eps/\delta$.\footnote{Explicitly, the combinatorial argument is as follows: In a chain of $N+1$ Ising spins, there are $\binom N k$ possibilities to position $k$ sign changes. This is equal to the number of possibilities to choose $k+1$ positive integers $x_1,\ldots,x_{k+1}$, such that $x_1+\cdots+x_{k+1}=N+1$. Now, if the distance between any two sign changes shall be larger than $m$, this is equivalent to requiring $x_2,\ldots,x_k> m$. By the change of variables $y_1=x_1$, $y_i=x_i-m$ for $i=2,\ldots,k$, $y_{k+1}=x_{k+1}$, we find the number of possibilities to be equal to the number of possibilities to choose $y_1,\ldots,y_{k+1}\in\IN$, such that $y_1+\ldots + y_{k+1}=N+1-(k-1)m$. Recalling the initial argument, this is $\binom{N-(k-1)m}{k}$. In our case we have $N=2\Ld(T)$ and $m=\lfloor\eps/\delta\rfloor$.} Taking into account that an element $\sigma\in\cS_{\Ld(T)}$ is uniquely determined by the choice of the value $\sigma_{\Ld(T)}\in\{\pm1\}$ and the position of its sign changes, we obtain
 	\begin{equation}\label{eq:combin}
 		\#\left\{\sigma\in\fd^{-1}(\Omega_\eps):n_\sigma=k\right\} = \begin{cases}\displaystyle 2\binom{2\Ld(T)}{k} & \mbox{for}\ k > \frac {2T}\eps, \\ \displaystyle2\binom{2\Ld(T)}{k} -2\binom{2\Ld(T)-(k-1)\lfloor\eps/\delta\rfloor}{k} & \mbox{for}\ 2\le k\le \frac {2T}\eps.	\end{cases}
 	\end{equation}
 	From the definition of nearest neighbor coupling, it easily follows that (cf. \cref{lemapp:isingnene} \cref{part:neneen})
 	\[E_{J_\jd,\Ld(T)}(\sigma) = 2\jd(n_\sigma-\Ld(T)) \qquad\mbox{for all}\ \sigma\in\cS_{\Ld(T)}.\]
 	Hence, combining \cref{eq:combin,eq:isingprob} and summing over all possible numbers of spin changes, we obtain
 	\begin{equation}\label{eq:prob}
 		\begin{aligned}
 			P_\delta(\Omega_\eps)
 			= &
 			\sum_{k=2}^{\lfloor\frac{2T}{\eps}\rfloor}\left(\binom{2\Ld(T)}{k} -\binom{2\Ld(T)-(k-1)\lfloor\eps/\delta\rfloor}{k}\right) \frac{2e^{2\jd(\Ld(T)-k)}}{Z_{J_\jd,\Ld(T)}} \\
 			& \qquad + \sum_{k=\lfloor\frac{2T}{\eps}\rfloor+1}^{2\Ld(T)}\binom{2\Ld(T)}{k}\frac{2e^{2\jd(\Ld(T)-k)}}{Z_{J_\jd,\Ld(T)}}.
 		\end{aligned}
 	\end{equation}
 	Since it is possible to explicitly calculate the partition function for nearest neighbor coupling (cf. \cref{lemapp:isingnene} \cref{part:nenepart}), we have
 	\begin{equation}\label{eq:boundby1}
 		\frac{2e^{2\jd\Ld(T)}}{Z_{J_\jd,\Ld(T)}} =  \left(\frac{e^\jd}{e^\jd+e^{-\jd}}\right)^{2\Ld(T)} < 1.
 	\end{equation}
 	Moreover, inserting the definition \cref{def:jd}, we have $e^{-2\jd k} = \delta^k$ and hence
 	\begin{equation}\label{eq:tailestimate}
 		\binom{2\Ld(T)}{k}e^{-2\jd k} \le \frac{(2\Ld(T))^k}{k!}\delta^k \le \frac{(2T+\delta)^k}{k!} \qquad\mbox{for all}\ k\le2\Ld(T),
 	\end{equation}
	 where we used $\Ld(T) = \lfloor\frac T\delta+\frac 12\rfloor \le \frac T\delta +\frac 12$ in the last step.
 	Along the same lines, we use Bernoulli's inequality to obtain for $k\le T/\eps$
 	\begin{equation}\label{eq:bernoulli}
 		\begin{aligned}
 			\left(\binom{2\Ld(T)}{k} -\binom{2\Ld(T)-(k-1)\lfloor\eps/\delta\rfloor}{k}\right)e^{-2\jd k}
 			&\le \frac{(2\Ld(T))^k - (2\Ld(T)-k(\eps/\delta+1))^k}{k!}\delta ^k\\
 			&\le \frac{(2\Ld(T))^k}{k!}\frac{k^2(\eps/\delta+1)}{2\Ld(T)}\delta^k \\&\le \frac{(2T+\delta)^{k-1}}{(k-1)!} k(\eps+\delta).
 		\end{aligned}
 	\end{equation}
	We can now insert \cref{eq:bernoulli,,eq:tailestimate,,eq:boundby1} into \cref{eq:prob}. Hence, for any $s_\eps\in[0,\frac 1\eps]$, we have
 	\[P_\delta(\Omega_\eps) \le \sum_{k=2}^{\lfloor Ts_\eps\rfloor} \frac{(2T+\delta)^{k-1}}{(k-1)!} k(\eps+\delta) + \sum_{k=\lfloor Ts_\eps\rfloor+1}^{2\Ld(T)}\frac{(2T+\delta)^k}{k!} \le Ts_\eps(\eps+\delta)e^{2T+\delta} + \sum_{k=\lfloor Ts_\eps\rfloor+1}^{\infty}\frac{(2T+\delta)^k}{k!},\]
 	where we estimated the first half of the first sum in \cref{eq:prob} by \cref{eq:bernoulli} and the second half using \cref{eq:tailestimate}.
 	Taking the limit $\delta\downarrow 0$, we observe
 	\[\limsup_{\delta\downarrow 0}P_\delta(\Omega_\eps)  \le Ts_\eps \eps e^{2T} + \sum_{k=\lfloor Ts_\eps\rfloor+1}^{\infty}\frac{(2T)^k}{k!}.\]
 	We choose $s_\eps$ such that both $\lim_{\eps\downarrow0}s_\eps = \infty$ and $\lim_{\eps\downarrow 0}s_\eps\eps =0$ hold, e.g., $s_\eps=\eps^{-1/2}$. Then, the summability of the second term proves \cref{eq:tightness} and hence $P_\delta$ weakly converges to $P_X$.
	
	Since $f$ is bounded and continuous, the statement for $N=0$ directly follows from the definition of weak convergence.
	Further, observe that for any fixed $N\in\IN$ and $t\in\IR^N$ the function $\omega\to\pi_t(\omega)f(\omega)$ is only discontinuous at those  $\omega$ having jumps exactly at the points given by the $N$-tuple $t$. Hence, the set of discontinuities has $P_X$-measure zero and the statement follows from the Portmanteau theorem \cite[Theorem 3.16]{Klenke.2020}.
\end{proof}
\noindent
We apply above lemma to prove the expectation value in \cref{prop:isinglimit} is a limit of expectation values in the nearest neighbor Ising model.
\begin{lemma}\label{lem:isinglimitnene}
	Assume $W:[-T,T]\to\IR$ is even and continuous and $\wde$ is as defined in \cref{def:wd}.\\
	For $N\in\IN_0$, let $-T\le t_1\le \cdots\le t_N \le T$. Then
	\begin{align*}
		\lim_{\delta\downarrow0}&\Bn{\sd{t_1}\cdots\sd{t_N}\exp\left(\sum_{i,j\in \Lambda_{\Ld(T)}}\wde_{\lvert i-j\lvert}\sigma_i\sigma_j\right)}_{\delta,T}
		\\&\qquad\qquad\qquad\qquad
		= \EE\left[X(t_1)\cdots X(t_N)\exp\left(\int_{-T}^{T}\int_{-T}^T W(t-s)X(s)X(t) ds dt\right)\right].
	\end{align*}
\end{lemma}
\begin{proof}
	We define $f_0:\cD_T\to\IR$ by
	\[f_0(\omega) = \omega(t_1)\cdots\omega(t_N)e^{\cI_0(\omega)},
	\qquad\mbox{where}\quad
	\cI_0(\omega) = \int_{-T}^T\int_{-T}^T W(t-s)\omega(t)\omega(s)dsdt.\]
	It is straight forward  to verify that $\cI_0:\cD_T\to\IR$ is bounded and continuous. Hence, we can apply \cref{lem:isingweakmeasure} and obtain
	\begin{equation}\label{eq:f0expbound}
		\lim_{\delta\downarrow0}\Bn{f_0(\fd(\sigma))}_{\delta,T} = \EE\left[ f_0 (X)\right].
	\end{equation}
	It remains to consider the left hand side and to  analyze $f_0(\fd(\sigma))$. 
Further, for $\sigma\in\cS_{\Ld(T)}$, we define
	\[g_\delta(\sigma) = \sd{t_1}\cdots\sd{t_N}e^{\mathcal{J}_\delta(\sigma)},
	\qquad\mbox{where}\quad
	\mathcal{J}_\delta(\sigma) = \sum_{i,j\in\Lambda_{\Ld(T)}}\wde_{\lvert i-j\lvert}\sigma_i\sigma_j.\]
	Since continuous functions on compact intervals are uniformly continuous, for any $\eps>0$ there exists a $\delta_\eps>0$, such that
	\[|W(t)-W(s)|<\eps\qquad\mbox{for any}\ t,s\in[-T,T]\ \mbox{with}\ |t-s|<\delta_\eps.\]
	Then, for any $\delta\in(0,\delta_\eps)$ and $\sigma\in\cS_{\Ld(T)}$, we use \cref{def:wd,def:id} to obtain
	\begin{align}
		\left|\mathcal{J}_\delta(\sigma)-  \cI_0(\fd(\sigma))\right| &= \left|\sum_{i,j\in\Lambda_{\Ld(T)}}\left[\wde_{\lvert i-j\lvert}\sigma_i\sigma_j - \int_{(i-\frac 12)\delta}^{(i+\frac 12)\delta}\int_{(j-\frac 12)\delta}^{(j+\frac 12)\delta} W(t-s)\sigma_i\sigma_jdsdt\right]\right|
		\nonumber \\& \le \sum_{i,j\in \Lambda_{\Ld(T)}} \delta^2 \sup \{|W(\delta t)-W(\delta |i-j|)|: t\in[|i-j|-1,|i-j|+1]\}
	\nonumber 	 \\&\le  (2\Ld(T)+1)^2\delta^2\eps \le 4 ( T+\delta)^2\eps.\label{eq:Iestimate}
	\end{align}
	Now, for all $\sigma\in\cS_{\Ld(T)}$ we have the algebraic identity
	\[ g_\delta(\sigma)-  f_0(\fd(\sigma)) =  f_0(\fd (\sigma))\left(e^{\mathcal{J}_\delta(\sigma)-  \cI_0(\fd(\sigma))}-1\right) .
\]
Using this identity and \eqref{eq:Iestimate}	 it follows that there exist  constants $C_1$ and $C_2$,  such that for $\eps > 0$ sufficiently small,  $\delta\in(0,\delta_\eps)$ and all  $\sigma\in\cS_{\Ld(T)}$ 
	\[\left|g_\delta(\sigma)-  f_0(\fd(\sigma))   \right| 
\le |  f_0(\fd(\sigma))     |C_1 \left|\mathcal{J}_\delta(\sigma)-  \cI_0(\fd(\sigma))  \right| \le C_2e^{4T^2\|W\|_\infty}(T+\delta)^2\eps. \]
	Since $\sigma\in\cS_{\Ld(T)}$ was arbitrary, this estimate also  holds for the expectation value, i.e.,
	\begin{equation}\label{eq:f0fdbound}
	\left|\Bn{g_\delta(\sigma)}_{\delta,T}-\Bn{ f_0(\fd(\sigma))     }_{\delta,T}\right|
	\le
	C_2e^{4T^2\|W\|_\infty}(T+\delta)^2\eps.
	\end{equation}
	Combining \cref{eq:f0fdbound,eq:f0expbound}, the statement follows.
\end{proof}
\noindent
It now remains to rewrite the Ising expectation value in above \lcnamecref{lem:isinglimitnene} as a correlation function.
\begin{proof}[\textbf{Proof of \cref{prop:isinglimit}}]
	By the definition \cref{def:expIsing}, we observe
	\[ \bw{\sd{t_1}\cdots\sd{t_N}}_{\delta,T} = \frac{\Bn{\sd{t_1}\cdots\sd{t_N}\exp\left(\sum\limits_{i,j\in\Lambda_{\Ld(T)}}\wde_{\lvert i-j\lvert}\sigma_i\sigma_j\right)}_{\delta,T}}{\Bn{\exp\left(\sum\limits_{i,j\in\Lambda_{\Ld(T)}}\wde_{\lvert i-j\lvert}\sigma_i\sigma_j\right)}_{\delta,T}}.  \]
	Hence, the statement follows from \cref{lem:isinglimitnene,def:partitionjump}
\end{proof}

\section{Proof of the Main Result}\label{sec:proof}

In this section we combine the central statements from the previous sections to  the proof of our main result \cref{main}. We  use the definitions from the previous \lcnamecref{sec:continuumlimit}.

\begin{proof}[\textbf{Proof of \cref{main}}] Fix $T>0$. Then using Fubini in the first equality and  \cref{prop:isinglimit}  in the second equality, we find
	\begin{align}
		\frac{1}{\cZ(W,T)}\EE&\left[\frac 1T \left(\int_{-T}^T X(t)dt\right)^2\exp\left(\int_{-T}^T\int_{-T}^T X(t)X(s)W(t-s)dtds\right)\right]\nonumber\\
		&=
		 \frac{1}{T\cZ(W,T)}\int_{-T}^T\int_{-T}^T \EE\left[X(u)X(v)\exp\left(\int_{-T}^T\int_{-T}^TX(t)X(s)W(t-s)dtds\right)\right]dudv\nonumber\\
		 &= \frac 1T \lim_{\delta\downarrow0} \int_{-T}^T\int_{-T}^T \bw{\sd u\sd v}_{\delta,T}dudv\nonumber\\
		 &=\lim_{\delta\downarrow0}\frac{1}{T}\sum_{i,j\in\Lambda_{\Ld(T)}}\delta^2\bw{\sigma_i\sigma_j}_{\delta,T} , \label{eq:finalargument}
	\end{align}
	where in the last step, we  calculated the integral  using that the integrand is a  step function.
To estimate \cref{eq:finalargument} we want to use  \cref{prop:corbound}.
First observe that 	by  definition \cref{def:jd} we find
	\begin{equation} \label{eq:ineqondelta}  \frac{1}{1-\tanh \jd} = \frac{e^{2\jd}+1}{2} < \frac 1\delta \qquad\mbox{for any}\ \delta\in(0,1).\end{equation}
	 Further, using the definition \cref{def:wd},  $W\in L^1(\IR)$ and the continuity of $W$, we have
	 \[\wde\in\ell^1 \qquad\mbox{and}\qquad \lim\limits_{\delta\downarrow0}\frac1\delta\|\wde\|_1 = \|W\|_1.\]
	Combining the above two relations, we find  for any constant $D>1$
	\begin{equation}  \label{eq:estfinalproof}  \frac{\|\wde\|_1}{1-\tanh \jd}\le D \|W\|_1 \qquad\mbox{if}\ \delta > 0 \ \mbox{is sufficiently small}.  \end{equation}
	Now let $\eps >0$ and $C_\eps$ be as in \cref{prop:corbound}. If $\|W\|_1\le \eps/D$ and $\delta > 0$ is sufficiently small, it follows from  \eqref{eq:estfinalproof}    that  the assumption \eqref{eq:condw}
of the proposition holds,   since   $\tanh x \leq x$  for $x \in [ 0,\infty)$. In that case   it hence follows from   \cref{prop:corbound} that
	\[ \sum_{i\in\IZ}\bw{\sigma_i\sigma_j}_{\delta,T}\le \frac{C}{1-\tanh(j_\delta)}\le \frac C\delta , \]
where  we used   \eqref{eq:ineqondelta} in the second inequality.
The last  displayed inequality implies
	\[ \sum_{i,j\in\Lambda_{\Ld(T)}}\bw{\sigma_i\sigma_j}_{\delta,T} \le \frac{C}{\delta}\Ld(T) . \]
	Inserting this into  \cref{eq:finalargument} and using $\Ld(T)\le \frac T\delta + \frac 12$ finishes the proof.
\end{proof}

\appendix
\section{Ising Model  with Nearest Neighbor Coupling}\label{app:isingnene}
In this appendix we consider  the Ising model with nearest neighbor coupling and calculate   the  known partition function and correlation functions. These
calculations are well-known and can be found in most textbooks covering the Ising model. Here  we use  slightly different notation than in the rest  of the paper.

We fix the lattice length $L\in\IN$, the lattice $\Lambda=\{-L,\ldots,L\}\subset \IZ$, the spin configuration space $\cS=\{\pm 1\}^{\Lambda}$, and the interaction strength $j\ge 0$.
The   Ising energy is defined by
\[ E(\sigma) = -\sum_{i=-L}^{L-1}j\sigma_i\sigma_{i+1} \qquad\mbox{for}\ \sigma\in\cS,  \]
the partition function by
\[ Z = \sum_{\sigma\in\cS}e^{-E(\sigma)} ,  \]
and the expectation value of $f:\cS\to\IR$ by
\[ \braket{f} = \frac{1}{Z}\sum_{\sigma\in\cS} f(\sigma)e^{-E(\sigma)}.  \]
We prove the following statements.
\begin{lemma}\label{lemapp:isingnene}\
	\begin{enumerate}[(i)]
		\item\label{part:neneen} For $\sigma\in\cS$ we write $n_\sigma = \#\{i=-L,\ldots,L-1:\sigma_i\sigma_{i+1}=-1\}$. Then
		$E(\sigma) = 2j(n_\sigma-L).$
		\item\label{part:nenepart} We have
			$ Z = 2\left(e^j+e^{-j}\right)^{2L}$.
		\item\label{part:nenecor}
			For $n\in\IN$ and $-L\le i_1\le \cdots \le i_n\le L$ we have
				\[\braket{\sigma_{i_1}\cdots\sigma_{i_n}} = \begin{cases} \displaystyle \tanh(j)^{\sum\limits_{k=1}^{N}|i_{2k}-{i_{2k-1}}|} & \mbox{if}\ n=2N,\\ 0 & \mbox{else}.  \end{cases}\]
	\end{enumerate}
\end{lemma}
\begin{proof}\cref{part:neneen} follows directly from the definition.
	For the proof of \cref{part:nenepart,part:nenecor}, we use the change of variables
	\[ \sigma_i' = \sigma_i\sigma_{i+1} \quad\mbox{for}\ i=-L,\ldots, L-1, \quad \sigma_L' = \sigma_L.\]
	Then, we have $\displaystyle E(\sigma)=-j\sum_{i=-L}^{L-1}\sigma_i' $ and hence
	\[ Z = \sum_{\sigma'\in \cS}\prod_{i=-L}^{L-1}e^{j\sigma_i '} = 2\prod_{i=-L}^{L-1}\sum_{\sigma_i'=\pm1}e^{j\sigma_i'} = 2\left(e^j+e^{-j}\right)^{2L},\]
	so \cref{part:nenepart} is proved.
	Now, if $-L\le i< j\le L$, we observe
	\[\sigma_i\sigma_j = (\sigma_i\sigma_{i+1})(\sigma_{i+1}\sigma_{i+2})\cdots(\sigma_{j-1}\sigma_j) = \sigma_i'\sigma_{i+1}'\cdots\sigma_{j-1}'.\]
	Assume $n=2N$. Then, we have
	\begin{align*}
		Z \braket{\sigma_{i_1}\cdots\sigma_{i_{2N}}} &= \sum_{\sigma\in \cS}\sigma_{i_1}\cdots\sigma_{i_{2N}}e^{-E(\sigma)}\\
		&= \sum_{\sigma\in\cS }e^{-E(\sigma)}\prod_{a=1}^N (\sigma_{i_{2a-1}}\sigma_{i_{2a-1}+1})\cdots (\sigma_{i_{2a}-1}\sigma_{i_{2a}})\\
		& = \sum_{\sigma'\in\cS}\prod_{a=1}^N \sigma_{i_{2a-1}}' \cdots \sigma_{i_{2a}-1}'\prod_{i=-L}^{L-1}e^{j\sigma_i'}.
	\end{align*}
	Now, for $l\in\Lambda$, we set \[s_l =  \begin{cases} 1 & \mbox{if there is some}\ a\in\IN\ \mbox{with}\ i_{2a-1}\le l< i_{2a},\\ 0 & \mbox{else}.  \end{cases}\] and
	$ S= \# \{l\in\Lambda:s_l=1\} = |i_2-i_1|+\cdots+|i_{2N}-i_{2N-1}| $. Inserting above, we obtain
	\begin{align*}
		 Z \braket{\sigma_{i_1}\cdots\sigma_{i_{2N}}} & = 2\prod_{l=-L}^{L-1}\sum_{\sigma_l'=\pm1} (\sigma_l')^{s_l}e^{j\sigma_l'}\\
		& = 2 (e^{j} - e^{-j} )^{S }   (e^{j} + e^{-j} )^{2L-S }.
	\end{align*}
	Combined with \cref{part:nenepart}, we obtain the identity \cref{part:nenecor} for even $n$. It remains to consider the case that $n$ is odd. The statement then, however, follows similar to the proof of \cref{lem:zeroexp} by the change of variables $\sigma\mapsto -\sigma$.
\end{proof}

\bibliographystyle{halpha-abbrv}
\bibliography{lit}

\end{document}